\documentclass[letterpaper, 10pt, conference]{ieeeconf}
\IEEEoverridecommandlockouts
\usepackage{graphicx,epstopdf,multirow,multicol,subfigure,booktabs,pgfplots}

\usepackage{amsmath}
\usepackage{amssymb,amsthm}
\usepackage{algorithm,algorithmic}
\usepackage{times}
\usepackage[hidelinks]{hyperref}
\usepackage{booktabs}
\allowdisplaybreaks
\newtheorem{proposition}{Proposition}
\newtheorem{corollary}{Corollary}

\title{\LARGE \bf
Formulations and algorithms for the multiple depot, fuel-constrained, multiple vehicle routing problem 
}

\author{Kaarthik Sundar$^{\dagger}$\thanks{$^{\dagger}$ Graduate Student, Dept. of Mechanical Engg., Texas A\&M University,
College Station, TX 77843. \texttt{kaarthiksundar@tamu.edu}},\;
Saravanan Venkatachalam$^{*}$\thanks{$^{*}$ Assistant Professor, Dept. of Industrial and Systems Engg., Wayne State University,
Detroit, MI 48202. \texttt{}},\;
Sivakumar Rathinam$^{\ddagger}$\thanks{$^{\ddagger}$ Assistant Professor, Dept. of Mechanical Engg., Texas A\&M University,
College Station, TX 77843.}} \;

\renewcommand{\baselinestretch}{1.05}

\begin{document}

\maketitle
\thispagestyle{empty}
\pagestyle{empty}

\begin{abstract}
We consider a multiple depot, multiple vehicle routing problem with fuel constraints. We are given a set of targets, a set of depots and a set of homogeneous vehicles, one for each depot. The depots are also allowed to act as refueling stations. The vehicles are allowed to refuel at any depot, and our objective is to determine a route for each vehicle with a minimum total cost such that each target is visited at least once by some vehicle, and the vehicles never run out fuel as it traverses its route. We refer this problem as Multiple Depot, Fuel-Constrained, Multiple Vehicle Routing Problem (FCMVRP). This paper presents four new mixed integer linear programming formulations to compute an optimal solution for the problem. Extensive computational results for a large set of instances are also presented.
\end{abstract}


\section{Introduction}
\label{sec:intro}
In this paper, we extend the classic multiple depot, multiple vehicle routing problem (MDMVRP) to include fuel constraints for the vehicles. We are given sets of targets, a set of depots, and a set of vehicles, with each vehicle initially stationed at a distinct depot. The depots also perform the role of refueling stations, and it is reasonable to assume that whenever a vehicle visits a depot, it refuels to its full capacity. Given this, the objective of FCMVRP is to determine a route for each vehicle starting and ending at its corresponding depot such that (i) each target is visited at least once by some vehicle, (ii) no vehicle runs out of fuel as it traverses its path, and (iii) the sum total cost of the routes for the vehicles is minimum. Some of the applications for the FCMVRP are path-planning for Unmanned Aerial Vehicles (UAVs) \cite{Sundar2012, Sundar2014, Levy2014}, routing for electric vehicles based on the locations of recharging stations \cite{Schneider2014, Hiermann2014}, and routing for green vehicles \cite{Erdougan2012}. Some of these application domains are elaborated on the following sections.


\subsection{Path-planning for UAVs \label{subsec:introuav}}
Small UAVs are being used routinely in military applications such as border patrol, reconnaissance, and surveillance expeditions, and civilian applications like remote sensing, traffic monitoring, and weather and hurricane monitoring \cite{Frew2009, Curry2004, ZajkowskiT2006}. Even though there are several advantages due to small platforms for UAVs, there are resource constraints due to their size and limited payload. It may not be possible for a small UAV to complete a surveillance mission before refueling at one of the depots due to the fuel constraints. For example, consider a typical surveillance mission with multiple vehicles each starting at a depot and together are required to monitor a set of targets. To complete this mission, the vehicles might have to start at their respective depot, then visit a subset of targets and reach one of the depots for refueling before starting a new route for the rest of the targets. This can be modeled as a FCMVRP with the depots acting as refueling stations. 

\subsection{Routing problem for green  and electric vehicles \label{subsec:introgvrp}}
Green vehicle routing problem is a variant of the Vehicle Routing Problem (VRP) and was introduced by authors in \cite{Erdougan2012} to account for the challenges associated with operating a fleet of alternate-fuel vehicles (AFVs). The US transportation sector accounts for 28\% of national greenhouse gas emissions \cite{USEPA}. Several efforts over many decades focusing towards the introduction of cleaner fuels (e.g. ultra low sulfur diesel) and efficient engine technologies have lead to reduced emissions and greater mileage per gallon of fuel used. Government organizations, municipalities, and private companies are converting their fleet of vehicles to AFVs either voluntarily to alleviate the environmental impact of fossil based fuels or to meet environmental regulations. For instance, FedEx, in its overseas operations, employs AFVs that run on biodiesel, liquid natural gas, or compressed natural gas. A major challenge that hinders the increase in usage of AFVs is the number of alternate-fuel stations available for refueling. The FCMVRP is a natural problem that arises in such a scenario. An algorithm to compute an optimal solution to the FCMVRP would generate low cost routes for the vehicles, while respecting their fuel constraints. 

Increasing concerns about climate changes and rising green house gas emissions drive the research in sustainable and energy efficient mobility. One such example is the introduction of electrically-powered vehicles. One of the main operational challenges for electric vehicles in transport applications is their limited range and the availability of recharging stations. The number of electric stations in the US is a mere 9,571 with a total of 24,631 charging outlets \cite{USDOE}. Fig. \ref{fig:map} shows a map with the locations of the electric stations in Texas, USA; observe that the distribution of the electric stations is very sparse except in the four major cities Dallas, Houston, Austin, and San Antonio. Successful adaption of electric vehicles will strongly depend on the methods alleviating the range and recharging limitations. If we consider the range and the recharging stations for the electric vehicles as analogues to the fuel capacity and refueling stations of vehicles that run on fossil-based or alternate fuels respectively, then the problem of electric vehicle routing subject to the range constraints and limited availability of electric stations can be modeled as an FCMVRP. Clearly, any feasible solution to the FCMVRP can be used to implement a feasible route for an electric vehicle. 

\begin{figure}
\centering
\includegraphics[scale=0.28]{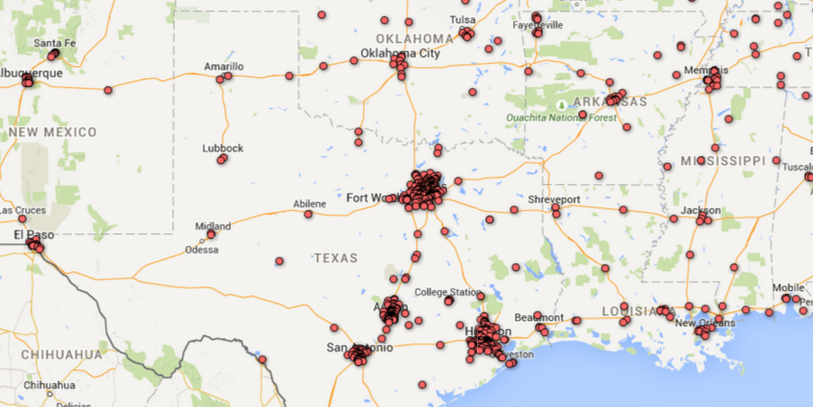}
\caption{Electric station locations in Texas, USA \cite{USDOE}}
 \label{fig:map}
\end{figure}

\section{Related work \label{sec:litreview}} 
The FCMVRP is NP-hard because it contains the VRP as a special case. The existing literature on the FCMVRP is quite scarce. The multiple depot, single vehicle variant of the FCMVRP was first introduced by authors in \cite{Khuller2007}. When the travel costs are symmetric and satisfy the triangle inequality, authors in \cite{Khuller2007} provide an approximation algorithm for this variant. They assume that the minimum fuel required to travel from any target to its nearest depot is at most equal to $F\alpha/2$ units, where $\alpha$ is a constant in the interval $[0,1)$ and $F$ is the fuel capacity of the vehicle. This is a reasonable assumption as, in any case, one cannot have a feasible tour if there is a target that cannot be visited from any of the depots. Using these assumptions, Khuller et al. \cite{Khuller2007} present a $(3(1+\alpha))/(2(1-\alpha))$ approximation algorithm for the problem. Authors in \cite{Sundar2012} formulate this multiple depot single vehicle variant as a mixed-integer linear program and present $k$-opt based exchange heuristics to obtain feasible solutions within $7\%$ of the optimal, on an average. Later, Sundar et al. \cite{Sundar2014} extend the approximation algorithm in $\cite{Khuller2007}$ to the asymmetric case and also present heuristics to solve the asymmetric version of this variant. Furthermore, variable neighborhood search heuristics for FCMVRP with heterogeneous vehicles, \emph{i.e.,} vehicles with different fuel capacities, are presented by Levy et al. \cite{Levy2014}. More recently, an approximation algorithm and heuristics are developed for the FCMVRP by the authors in \cite{Mitchell2015}.

Variants of the classic VRP that are closely related to the FCMVRP include the distance constrained VRP \cite{Laporte1984, Li1992, Kara2010, Kara2011, Nagarajan2012}, the orienteering problem \cite{Fischetti1998, Vansteenwegen2011}, and the capacitated version of the arc routing problem \cite{Ghiani2004, Polacek2008}. The distance constrained VRP is a special case of the FCMVRP with a single vehicle and single depot that can be considered as a fuel station. The FCMVRP is also quite different and more general compared to orienteering problem where one is interested in maximizing the number of targets visited by the vehicle subject to its fuel constraints. Lastly, the arc routing problem is a single depot VRP given a set of intermediate facilities, and the vehicle has to cover a subset of edges along which targets are present. The vehicle is required to collect goods from the targets as it traverses the given set of edges and unloads the goods at the intermediate facilities. The goal of this problem is to find a tour of minimum length that starts and ends at the depot such that the vehicle visits the given subset of edges, and the total amount of goods carried by the vehicle does not exceed the capacity of the vehicle along the tour. One of the key differences between the arc routing problem and the FCMVRP is that there is no requirement that any subset of edges must be visited in the FCMVRP.

The aim of this paper is to introduce and compare four different formulations for the FCMVRP and present branch-and-cut algorithms for the formulations. The first two formulations are arc-based, and the rest are node-based formulations that use the Miller-Tucker-Zemlin (MTZ) constraints \cite{MTZ1960}. The major contributions of this paper are as follow: (1) present four new formulations for the FCMVRP, (2) compare the formulations both analytically and empirically, and (3) through extensive computational experiments, show that instances with maximum of 40 targets are within the computational reach of a branch-and-cut algorithm based on the best of the four formulations. 

The rest of the paper is organized as follows. Sec. \ref{sec:definition} states the formal definition of the problem and introduces notations. In Sec. \ref{sec:formulation}, we develop the four mixed integer linear programming formulations. The first two formulations are arc-based and the rest are node-based formulations \emph{i.e.,} decision variables for enforcing the fuel constraints are introduced for each edge and each target for the arc-based and the node-based formulations, respectively. The linear programming relaxations of the formulations are analytically compared in this section. In Sec. \ref{sec:results}, we present the computational results followed by conclusions and possible extensions. 

\section{Problem definition \label{sec:definition}}
Let $T$ denote the set of targets $\{t_1,\dots,t_n\}$ . Let $D$ denote the set of depots or refueling stations $\{d_1,\dots, d_k\}$; each depot $d_k$ is equipped with a vehicle $v_k$. The FCMVRP is defined on a directed graph $G=(V,E)$ where $V=T\cup D$ and $E$ is the set of edges joining any two vertices in $V$. We assume that $G$ does not contain any self-loops. Each edge $(i,j) \in E$ is associated with a non-negative cost $c_{ij}$ required to travel from vertex $i$ to vertex $j$ and $f_{ij}$, the fuel spent by traveling from $i$ to $j$. It is assumed that the cost of traveling from vertex $i$ to vertex $j$ is directly proportional to the fuel spent in traversing the edge $(i,j)$ \emph{i.e.}, $c_{ij} = K\cdot f_{ij}$ ($c_{ij}$ and $c_{ji}$ may be different, but for the purpose of this paper, we assume $c_{ij} = c_{ji}$). It is also assumed that travel costs satisfy the triangle inequality \emph{i.e.}, for every $i,j,k\in V$, $c_{ij} + c_{jk} \geq c_{ik}$. Furthermore, let $F$ denote the fuel capacity of all the vehicles. The FCMVRP consists of finding a route for each vehicle such that the vehicle $v_k$ starts and ends its route at its depot $d_k$, each target is visited at least once by some vehicle, the fuel required by any vehicle to travel any segment of the route which joins two consecutive depots in the route must be at most equal to $F$, and the sum of the cost of all the edges present in the routes is a minimum.

\section{Mathematical formulations \label{sec:formulation}}
This section presents four formulations for the FCMVRP. The first two formulations are arc based, and the remaining formulations are node based. The arc based and edge based formulations have additional decision variables for each edge and vertex respectively, to impose the fuel constraints. 
%
For any given formulation $\mathcal F$, let $\mathcal F^L$ denote its linear programming relaxation obtained by allowing the integer variables to take continuous values within the lower and upper integer bounds, and $\operatorname{opt}(\mathcal F)$ denote the cost of its optimal solution. 

\subsection{Arc-based formulations \label{subsec:arcbased}}
We first present an arc based formulation $\mathcal F_1$ for the FCMVRP, inspired by the models for standard routing problems \cite{Toth2001, Kara2011}. Each edge $(i,j)\in E$ is associated with a variable $x_{ij}$, which equals $1$ if the edge $(i,j)$ is traversed by the vehicle, and $0$ otherwise. Also, associated with each edge $(i,j)$ is a flow variable $z_{ij}$ which denotes the total fuel consumed by any vehicle as it starts from a depot to the vertex $j$, when the predecessor of $j$ is $i$. Using the above variables, the formulation $\mathcal F_1$ is given as follows:
\begin{flalign}
&(\mathcal F_1) \quad \text{Minimize} \quad \sum_{(i,j)\in E} c_{ij} x_{ij} \notag& \\
&\text{subject to:} \notag & \\
&\sum_{i\in V} x_{di} = \sum_{i\in V} x_{id} \quad \forall \, d\in D,\label{eq:f1-degree-d}& \\
&\sum_{i\in V} x_{ij} = 1 \text{ and } \sum_{i\in V} x_{ji} = 1 \quad \forall \, j \in T, \label{eq:f1-degree-t} &\\
&\sum_{j\in V} z_{ij} - \sum_{j\in V} z_{ji} = \sum_{j\in V} f_{ij} x_{ij} \quad \forall \, i\in T, \label{eq:f1-fuel} &\\
&0 \leq z_{ij} \leq Fx_{ij} \quad \forall \, (i,j) \in E, \label{eq:f1-fuel2} &\\
&z_{di} = f_{di}x_{di} \quad \forall \, i\in T, \, d \in D \label{eq:f1-fuel3},  \text{ and}  &\\
&x_{ij} \in \{0,1\} \quad \forall \, (i,j) \in E \label{eq:f1-integer}.& 
\end{flalign}
In the above formulation the Eqs. \eqref{eq:f1-degree-d} -- \eqref{eq:f1-degree-t} impose the degree constraints on the depots and the targets. The constraints in \eqref{eq:f1-fuel} are the connectivity constraints; they eliminate sub tours of the targets. \eqref{eq:f1-fuel2} and \eqref{eq:f1-fuel3} together impose $0\leq z_{ij} \leq F$ and they ensure that the fuel consumed by the vehicle to travel up to a depot does not exceed the fuel capacity $F$. Finally, the constraints in Eqs. \eqref{eq:f1-integer} impose the binary restrictions on the variables.   

Now, we present another arc-based formulation $\mathcal F_2$ which is a strengthened version of $\mathcal F_1$. The following proposition is a modified version of the Prop. 1 presented in \cite{Kara2011} for the distance constrained vehicle routing problem; it strengthens the bounds given by the constraints in \eqref{eq:f1-fuel2}. 
\begin{proposition} \label{prop:strengthen} The inequalities in \eqref{eq:f1-fuel2} can be strengthened as follows:
\begin{flalign}
& z_{ij} \leq (F - t_j)x_{ij} \quad \forall j\in T,\, (i,j) \in E \label{eq:f2-fuel1} & \\
& z_{id} \leq Fx_{id} \quad \forall i \in T \text{ and } d \in D \label{eq:f2-fuel2} & \\
& z_{ij} \geq (s_i + f_{ij}) x_{ij} \quad \forall i\in T, \, (i,j) \in E \label{eq:f2-fuel3} &
\end{flalign}
where, $t_i = \min_{d\in D} f_{id}$ and $s_i = \min_{d\in D} f_{di}$.
\end{proposition} 
\begin{proof}
When $j$ is a depot, the constraints in \eqref{eq:f2-fuel2} and \eqref{eq:f1-fuel2} coincide. We now discuss the case when both $i$ and $j$ are targets. When $x_{ij} = 1$, any vehicle that traverses this edge $(i,j)$ consumes at least $(s_i + f_{ij})$ amount of fuel. As a result, the constraint in \eqref{eq:f2-fuel3} strengthens the lower bound of $z_{ij}$ in \eqref{eq:f1-fuel2}. Similarly, the total fuel consumed by any vehicle that traverses the edge $(i,j)$ cannot be greater that $(F - t_j)$, where $t_j$ is the minimum amount of fuel required by any vehicle to reach a depot from target $j$. Therefore, the constraint in \eqref{eq:f2-fuel1} strengthens the upper bound of $z_{ij}$ in \eqref{eq:f1-fuel2}.
\end{proof}
Hence, the second arc-based formulation is as follows:
\begin{flalign}
&(\mathcal F_2) \quad \text{Minimize} \quad \sum_{(i,j)\in E} c_{ij} x_{ij} \notag& \\
&\text{subject to: \eqref{eq:f1-degree-d} -- \eqref{eq:f1-fuel}, \eqref{eq:f1-fuel3} -- \eqref{eq:f1-integer}, and \eqref{eq:f2-fuel1} -- \eqref{eq:f2-fuel3}.} \notag & 
\end{flalign}
\begin{corollary} \label{cor:LP} $\operatorname{opt}(\mathcal F_2^L) \geq \operatorname{opt}(\mathcal F_1^L)$. \hfill \qed
\end{corollary}

\subsection{Node-based formulations \label{subsec:nodebased}} 
In this section, we present a node-based formulation for the FCMVRP based on the models for the distance constrained VRP in \cite{Desrochers1991, Kara2010}. For the node based formulation, apart from the binary variable $x_{ij}$ for each edge $(i,j) \in E$, we have an auxiliary variable $u_i$ for each vertex $i$, that indicates the amount of fuel spent by a vehicle when it reaches the vertex $i$. We assume $u_d = 0$ as the vehicles are refueled to their capacity when they reach a depot. In addition, we will also use the following two parameters: $t_i = \min_{d\in D} f_{id}$ and $s_i = \min_{d\in D} f_{di}$ for every vertex $i \in V$. For any $d \in D$, $t_d = 0$ and $s_d = 0$. Using the above notations, the formulation $\mathcal F_3$ is given as follows:
\begin{flalign}
&(\mathcal F_3) \quad \text{Minimize} \quad \sum_{(i,j)\in E} c_{ij} x_{ij} \notag& \\
&\text{subject to: \eqref{eq:f1-degree-d}, \eqref{eq:f1-degree-t}, and \eqref{eq:f1-integer}}, \notag & \\
& u_i - u_j + M x_{ij} \leq M - f_{ij} \quad \forall i \in V, j \in T, \label{eq:f3-fuel1}& \\
& u_i \geq s_i + \sum_{d\in D} (f_{di} - s_i)x_{di} \quad \forall i \in T \label{eq:f3-fuel2},  \text{ and} & \\
& u_i \leq F - t_i - \sum_{d \in D} (f_{id}-t_i)x_{id} \quad \forall i \in T. \label{eq:f3-fuel3} & 
\end{flalign}
The constraint in Eq. \eqref{eq:f3-fuel1} serves both as sub tour elimination and fuel constraints. It eliminates sub tours of the targets and ensures any route that starts and ends at a depot consumes at most $F$ amount of fuel. This can be easily observed by aggregating the constraints for any sub tour of the targets and for any route starting and ending at a depot \cite{Desrochers1991}. The value of $M$ in the constraint is given by $M = \max_{(i,j)\in E} \{F - s_j -t_i + f_{ij}\}$. The constraints in Eqs. \eqref{eq:f3-fuel2} and \eqref{eq:f3-fuel3} specify the upper and lower bounds on $u_i$, for every vertex $i$. The following proposition strengthens the fuel constraints and the bounds on $u_i$. 

\begin{proposition} \label{prop:lifting} The inequalities in \eqref{eq:f3-fuel1}, \eqref{eq:f3-fuel2}, and \eqref{eq:f3-fuel3} can be strengthened as follows:
\begin{flalign}
& u_i - u_j + M x_{ij} + (M-f_{ij}-f_{ji})x_{ji} \leq M - f_{ij} \notag & \\
& \hspace{35ex} \forall i,j \in T \label{eq:f3-fuel4}, &\\
& u_i \geq \sum_{j \in V} (s_j + f_{ji}) x_{ji} \quad \forall i \in T \label{eq:f3-fuel5}, & \\
& u_i \leq F - \sum_{j \in V} (t_j + f_{ij}) x_{ij} \quad \forall i \in T, \text{ and} \label{eq:f3-fuel6} & \\
& u_i \leq F - t_i - \sum_{d\in D} (F - t_i - f_{di}) x_{di} \quad \forall i \in T. \label{eq:f3-fuel7} &
\end{flalign}
where, $x_{ii} = 0$ and $x_{ij} = 0$ whenever $s_i + f_{ij} + t_j > F$.
\end{proposition} 
\begin{proof}
The constraint in Eq. \eqref{eq:f3-fuel4} can be obtained by lifting the variable $x_{ji}$ in Eq. \eqref{eq:f3-fuel1}. A constraint is said to be ``valid'' if it does not remove any feasible solution to the FCMVRP. We compute the value of the coefficient $\alpha$ that makes the following constraint valid: $$u_i - u_j + M x_{ij} + \alpha x_{ji} \leq M - f_{ij}.$$ The equation is valid when $x_{ji} = 0$, as it reduces to \eqref{eq:f3-fuel1}. When $x_{ji} = 1$, we have $x_{ij} = 0$ and $u_j + f_{ji} = u_i$. Hence, the best value of $\alpha$ that makes the equation valid is given by $M - f_{ij} - f_{ji}$. 

Similarly, Eq. \eqref{eq:f3-fuel5} can be obtained by lifting every $x_{ji}$ variable for $j\in T$ in any order. We will illustrate the lifting procedure for one of the $x_{ji}$ variables. This involves computing the coefficient $\alpha$ that makes the following constraint valid: $$u_i \geq s_i + \sum_{d\in D} (f_{di} - s_i)x_{di} + \alpha x_{ji}.$$ The above equation is valid when $x_{ji} = 0$, and when $x_{ji} = 1$, we have $x_{di} = 0$ and $\alpha \leq u_i - s_i$. The best value of $\alpha$ that does not remove any feasible FCMVRP is hence given by $s_j + f_{ji} - s_i$. Similarly, the coefficients of the other $x_{ji}$ variables can be computed. The resulting constraint is given by $$u_i \geq s_i + \sum_{j \in V} (s_j + f_{ji}-s_i) x_{ji} \quad \forall i \in V.$$ In the above equation, $s_j = 0$ for $j\in D$. The above equation reduces to Eq. \eqref{eq:f3-fuel5} due to the degree constraints in \eqref{eq:f1-degree-t}. The constraints in Eq. \eqref{eq:f3-fuel6} are similarly obtained from \eqref{eq:f3-fuel3} by lifting the $x_{ij}$ variable for every $j\in T$. The proof is omitted as it is similar to the previous ones in the proposition. The constraints in Eq. \eqref{eq:f3-fuel7} are valid bounding constraints for the FCMVRP when the target $i$ is the first target that is visited by any vehicle as it leaves the depot. In this case, the Eq. \eqref{eq:f3-fuel3} reduces to $u_i \leq F-t_i$. We further strengthen this constraint by lifting the variable $x_{di}$ for every $d\in D$. The lifting coefficient $\alpha$ for $x_{di}$ takes the value $-(F-t_i-f_{di})$ and the resulting constraint is given by Eq. \eqref{eq:f3-fuel7}.
\end{proof}

Hence, the second node-based formulation is as follows:
\begin{flalign}
&(\mathcal F_4) \quad \text{Minimize} \quad \sum_{(i,j)\in E} c_{ij} x_{ij} \notag& \\
&\text{subject to: \eqref{eq:f1-degree-d}, \eqref{eq:f1-degree-t}, \eqref{eq:f1-integer}, and \eqref{eq:f3-fuel4} -- \eqref{eq:f3-fuel7}} \notag. & 
\end{flalign}

\begin{corollary} \label{cor:LPnode} $\operatorname{opt}(\mathcal F_4^L) \geq \operatorname{opt}(\mathcal F_3^L)$. \hfill \qed
\end{corollary}

\section{Computational results} \label{sec:results}
In this section, we discuss the computational performance of the four formulations presented in the previous section. The mixed integer linear programs were implemented in Java, using the traditional branch-and-cut framework of CPLEX version 12.4. All the simulations were performed on a Dell Precision T5500 workstation (Intel Xeon E5630 processor @2.53 GHz, 12 GB RAM). The computation times reported are expressed in seconds, and we imposed a time limit of 3,600 seconds for each run of the algorithm. The performance of the algorithm was tested with randomly generated test instances. \\

\noindent {\it Instance generation}

The problem instances were randomly generated in a square grid of size [100,100] with 5 fixed depot locations. The number of targets varies from $10$ to $40$ in steps of five, while their locations were uniformly distributed in the square grid; for each $|T| \in \{10,15,20,25,30,25,40\}$, we generate five random instances. Each depot contains a vehicle. The travel costs and the fuel consumed to travel between any pair of vertices are assumed to be equal to the euclidean distances between the pair. For each of these problems, we generate four possible fuel capacities $F$ as a function of the the distance to the farthest target from any depot $\lambda$. The fuel capacity $F$ of the vehicles gets the values $2.25\lambda$, $2.5\lambda$, $2.75\lambda$ and $3\lambda$. In total, we generate $140$ instances and run the branch-and-cut algorithm for all the formulations.  \\

Tables \ref{tab:1} and \ref{tab:2}, and Fig. \ref{fig:times}--\ref{fig:LB} summarize the computational behavior of the algorithms for all the $140$ instances. The following nomenclature is used throughout the rest of the paper:\medskip{}

\noindent $\#$: instance number;

\noindent $\operatorname{opt}(\mathcal F_i^L)$: linear programming relaxation solution for formulation $i$;

\noindent $n$: instance size \emph{i.e.}, number of targets in the instance;

\noindent \%-LB: percentage LB/opt, where LB is the objective value of the linear programming relaxation computed at the root node of the branch and bound tree and opt is the cost of the optimal solution to the instance;

\noindent total: total number of test instances of a given size;

\noindent succ: number of instances for which optimal solutions were computed within a time limit of 3,600 seconds.  \\

\noindent Table \ref{tab:1} compares the cost of the linear programming (LP) relaxations of the four formulations presented in Sec. \ref{sec:formulation} for the 40 target instances.  The results in table \ref{tab:1} provide an empirical comparison of the formulations presented in \ref{sec:formulation}; the observed behavior is expected because the formulations $\mathcal F_2$ and $\mathcal F_4$ are strengthened versions of $\mathcal F_1$ and $\mathcal F_3$, respectively (see corollaries \ref{cor:LP} and \ref{cor:LPnode}). As for the LP relaxations of formulations $\mathcal F_2$ and $\mathcal F_4$, it is difficult to conclude that one is better than the other since $\mathcal F_4$ produces better relaxation values than $\mathcal F_2$ only for 60\% of the instances. Hence, the rest of the computational results compares the formulations $\mathcal F_2$ and $\mathcal F_4$.  

\begin{table}
\centering
\caption{Cost of the LP relaxation for the 40 target instances.}
\label{tab:1}
\begin{tabular}{lrrrr} 
\toprule 
$\#$ & $\operatorname{opt}(\mathcal F_1^L)$ & $\operatorname{opt}(\mathcal F_2^L)$ & $\operatorname{opt}(\mathcal F_3^L)$ & $\operatorname{opt}(\mathcal F_4^L)$ \\ 
\midrule 
1 & 496.42 & 509.24 & 426.17 & 518.00 \\
2 &	487.31 & 496.39 & 426.17 & 518.00 \\
3 &	480.55 & 487.40 & 426.17 & 518.00 \\ 
4 &	475.23 & 480.33 & 426.17 & 518.00 \\
5 &	444.35 & 458.01 & 389.08 & 434.00 \\ 
6 &	435.45 & 445.70 & 389.08 & 434.00 \\ 
7 &	428.44 & 436.47 & 389.08 & 434.00 \\
8 &	423.06 & 429.97 & 389.08 & 434.00 \\ 
9 &	396.10 & 403.96 & 367.11 & 452.00 \\ 
10 & 392.87 & 398.72 & 367.11 & 452.00 \\
11 & 390.42 & 394.66 & 367.11 & 452.00 \\
12 & 388.40 & 391.85 & 367.11 & 452.00 \\ 
13 & 481.22 & 493.64 & 427.04 & 461.00 \\
14 & 469.76 & 479.81 & 427.04 & 461.00 \\
15 & 461.16 & 469.20 & 427.04 & 461.00 \\
16 & 454.80 & 461.47 & 427.04 & 461.00 \\
17 & 503.19 & 516.58 & 461.07 & 523.00 \\
18 & 494.98 & 504.84 & 461.07 & 523.00 \\
19 & 489.64 & 496.31 & 461.07 & 523.00 \\ 
20 & 485.92 & 489.99 & 461.07 & 523.00 \\
\bottomrule 
\end{tabular} 
\end{table}

Table \ref{tab:2} shows the number of instances of different sizes solved to optimality by the formulations $\mathcal F_2$ and $\mathcal F_4$ within the time limit of 3600 seconds. The plot in Fig. \ref{fig:times} shows the average time taken by the two formulations to compute the optimal solution to the FCMVRP. The table \ref{tab:2} and Fig. \ref{fig:times} indicate that the arc-based formulation $\mathcal F_2$ outperforms the node-based formulation $\mathcal F_4$ for the larger instances. For the smaller sized instances, it is difficult to differentiate between the two formulations. The plot in Fig. \ref{fig:LB} shows the percentage LB/opt for both the formulations (LB is the objective value of the linear programming relaxation computed at the root node of the branch and bound tree and opt is the cost of the optimal solution to the instance; for the instances not solved to optimality, opt represents the cost of the best feasible solution obtained at the end of 3,600 seconds). We observe that the \%LB is consistently better for formulation $\mathcal F_2$. This plot also provides empirical evidence to the claim that the arc based formulation $\mathcal F_2$ outperforms the node based formulation $\mathcal F_4$. 

\begin{table}
\centering
\caption{Comparison of formulations $\mathcal F_2$ and $\mathcal F_4$.}
\label{tab:2}
\begin{tabular}{lrrr} 
\toprule 
 & & $\mathcal F_2$ & $\mathcal F_4$ \\
 \cline{3-4}
 $n$ & total & succ & succ \\
\midrule 
10 & 20 & 20 & 20\\
15 & 20 & 20 & 20\\
20 & 20 & 20 & 20\\
25 & 20 & 20 & 14\\ 
30 & 20 & 20 & 5\\
35 & 20 & 20 & 15\\ 
40 & 20 & 19 & 1\\
\bottomrule 
\end{tabular} 
\end{table}

\begin{figure}
\begin{tikzpicture}
\begin{axis}[
	x tick label style={
		/pgf/number format/1000 sep=},
	ylabel=Time (seconds),
	xlabel=$n$,
	enlargelimits=0.05,
	legend style={at={(0.2,0.95)}, draw=none,
	anchor=north,legend columns=-1},
	ybar interval=0.5,
]
\addplot 
	coordinates {(15,0) (20,7) (25,32) (30,42) (35,158) (40,709) (45,709)};
\addplot 
	coordinates {(15,1) (20,291) (25,1826) (30,3236) (35,923) (40,3470) (45,709)};
\legend{$\mathcal F_2$,$\mathcal F_4$}
\end{axis}
\end{tikzpicture}
\caption{Average time taken to compute the optimal solution.}
\label{fig:times}
\end{figure}
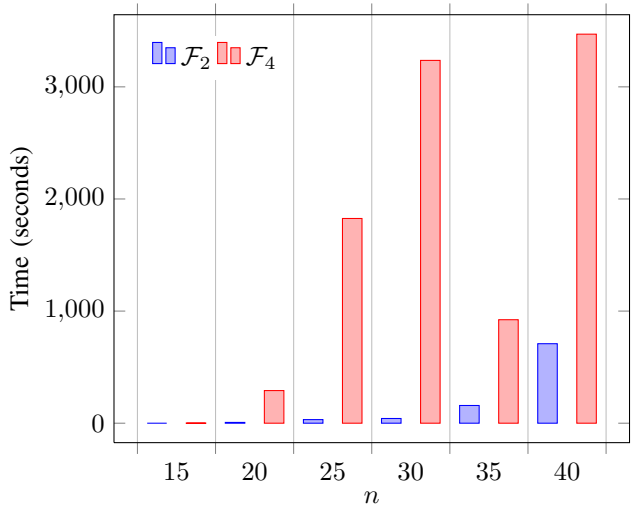

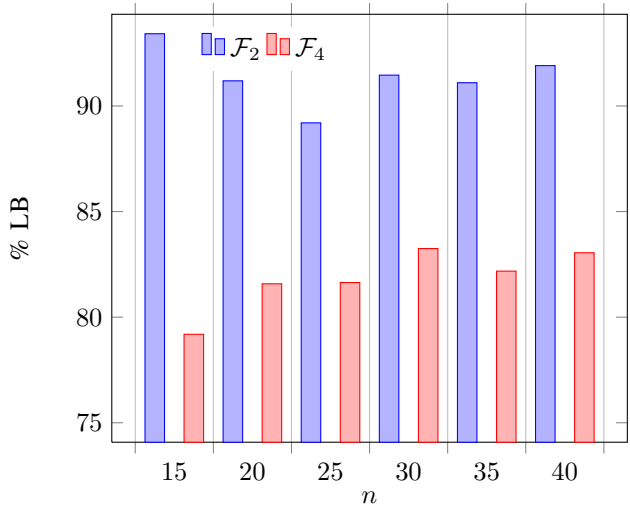
\begin{figure}
\begin{tikzpicture}
\begin{axis}[
	x tick label style={
		/pgf/number format/1000 sep=},
	ylabel=\% LB,
	xlabel=$n$,
	enlargelimits=0.05,
	legend style={at={(0.3,0.97)},draw=none,
	anchor=north,legend columns=-1},
	ybar interval=0.5,
]
\addplot 
	coordinates {(15,93.42) (20,91.19) (25,89.20) (30,91.46) (35,91.10) (40,91.91) (45,75)};
\addplot 
	coordinates {(15,79.19) (20,81.58) (25,81.64) (30,83.25) (35,82.18) (40,83.05) (45,75)};
\legend{$\mathcal F_2$,$\mathcal F_4$}
\end{axis}
\end{tikzpicture}
\caption{Average \% LB.}
\label{fig:LB}
\end{figure}

\section{Conclusions and future work \label{sec:conclusion}}
In this paper, we have presented four different mixed integer linear programming formulations for the multiple depot fuel constrained multiple vehicle routing problem. The problem arises frequently in the context of path planning for UAVs, green vehicle routing and routing electric vehicles. The formulations have been compared both analytically and empirically, and it is observed that a strengthened arc based formulation ($\mathcal F_2$) performs the best in terms of computing optimal solutions to the problem. Computational experiments on a large number of test instances corroborates this observation. Future work can be directed towards developing similar mixed integer linear programming formulations and branch-and-cut algorithms to solve a heterogeneous variant of the problem \emph{i.e.,} with vehicles having different fuel capacities. 

\bibliographystyle{IEEEtran}
\bibliography{references}

\end{document}